\documentclass{article}
\usepackage[utf8]{inputenc}
\usepackage{indentfirst}
\usepackage{graphicx,tikz}
\usetikzlibrary{decorations.pathreplacing}
\usetikzlibrary{arrows}
\usepackage{fancyhdr}
\usepackage{hyperref}
\hypersetup{pdftitle={Relative entropy in CFT},pdfauthor={Lorenzo Panebianco}}
\graphicspath{ {images/} }
\usepackage[T1]{fontenc}
\usepackage{mathrsfs}
\usepackage[english]{babel}
\usepackage{float,floatflt, epsfig}
\usepackage{subfig}
\usepackage{color}
\usepackage{caption}
\usepackage{fancyhdr}
\usepackage{lipsum} 
\usepackage{latexsym}
\usepackage{amssymb,amsthm,amsmath}
\usepackage{bbm}
\usepackage{amscd,xypic}
\usepackage{mathrsfs}
\usepackage{braket}
\usepackage[toc,page]{appendix}
\usepackage{enumerate} 
\usepackage{booktabs}
\usepackage{mathtools}
\usepackage[all]{xy}
\usepackage[margin=1.5in]{geometry}

\newcommand{\R}{\mathbb{R}}
\newcommand{\Z}{\mathbb{Z}}
\newcommand{\Diff}{\textit{\em Diff}_+(S^1)}

\newcommand{\hi}{\mathcal{H}}
\newcommand{\al}{\mathcal{A}}

\newcommand{\id}{\mathbbm{1} }
\newcommand{\M}{\mathcal{M}}

\theoremstyle{plain}
\newtheorem{thm}{Theorem}

\newtheorem*{thm*}{Theorem}
\newtheorem{cor}[thm]{Corollary}

\newtheorem{lem}[thm]{Lemma} 
\newtheorem{prop}[thm]{Proposition}

\newtheorem*{claim*}{Claim} 

\theoremstyle{definition}
\newtheorem{defn}[thm]{Definition}

\newtheorem*{defn*}{Definition}

\theoremstyle{remark}
\newtheorem{remark}[thm]{Remark}

\theoremstyle{definition}

\theoremstyle{remark}

\title{{A formula for the relative entropy in chiral CFT}}
\author{Lorenzo Panebianco* \\ \\
 *Dipartimento di Matematica,  Università di Roma “La Sapienza” \\
Piazzale Aldo Moro 5, 00185–Roma, Italy \\
E-mail: panebianco@mat.uniroma1.it}
\date{November 2019}

\begin{document}

\maketitle

\begin{abstract}
   We prove the QNEC on the Virasoro nets for a class of unitary states extending the coherent states, that is states obtained by applying an exponentiated stress energy tensor to the vacuum. We also verify the Bekenstein Bound by computing the relative entropy on a bounded interval.
\end{abstract}

\section{Introduction}

In this work we extend the recent results of \cite{Hollands} and we prove the Quantum Null Energy Condition (QNEC) for coherent states in (1+1)-dimensional chiral Conformal Field Theory (CFT) by explicitly computing the vacuum relative entropy. \\

The first non commutative entropy notion, von Neumann’s quantum entropy, was originally designed as a Quantum Mechanics version of Shannon’s entropy: if a state $\psi$ has density matrix $\rho_\psi$ then the von Neumann entropy is given by
\[
S_\psi = - \textit{\em tr}(\rho_\psi \log \rho_\psi) \,.
\]
However, in Quantum Field Theory local von Neumann algebras are typically factors of type $III_1$ (see \cite{Longo}), no trace or density matrix exists and the von Neumann entropy is undefined. Nonetheless, the Tomita-Takesaki modular theory applies and one may consider the Araki relative entropy \cite{Araki}
\[
S(\varphi \Vert \psi)  = -(\xi|\log \Delta_{\eta, \xi} \xi) \,.
\]
Here $\xi$ and $\eta$ are standard vectors of a von Neumann algebra $\M$ and $\Delta_{\eta, \xi}$ is the relative modular operator. The quantity $S(\varphi \Vert \psi )$ measures how $\varphi$ deviates from $\psi$. From the information theoretical viewpoint, $S(\varphi \Vert \psi)$ is the mean value in the state $\varphi$ of the difference between the information carried by the state $\psi$ and the state $\varphi$. \\

In \cite{Hollands}, the relative entropy is applied in (1+1)-dimensional chiral CFT as follows: $ \M $ is the local algebra $ \al (0, +\infty) $, $\varphi$ is the vacuum state $\omega$ given by the vacuum vector $\Omega$ and $\psi$ is a coherent state $\omega_f$ given by the vector $e^{iT(f)} \Omega$, with $f$ a real smooth vector field with compact support on the real line and $T(f)$ the stress-energy tensor. If $V$ is the unitary projective representation of $\Diff$ and $\rho $ is the exponential of the vector field on the circle $C_*f$, with $C$ the Cayley transform,  then one has that $\omega_f = \omega_{V(\rho)}$, that is $\omega_f$ is represented by the vector $V(\rho)\Omega$. In \cite{Hollands} it is proved that if $f(0)=0$ then we have
\begin{equation*}
\begin{split}
    S_{(0, +\infty)}(\omega_{V(\rho)} \Vert \omega) = \frac{c}{24}\int_{0}^{+ \infty} u \left( \frac{\eta''(u)}{\eta'(u)}\right)^2 du \,,
\end{split}
\end{equation*}
where $\eta$ is the inverse of the diffeomorphism $\rho$. In this work we  remove the condition $f(0)=0$. By doing this, we are able to prove that
\[
        S_{(t, +\infty)}(\omega_{V(\rho)} \Vert \omega)   = \frac{c}{24}\int_{t}^{+ \infty} (u-t) \left( \frac{\eta''(t)}{\eta'(t)}\right)^2  du \,.
\]
More in general we notice that the same expression holds if $\rho$ is a generic diffeomorphism of the circle fixing $-1$ and with unitary derivative in such point. This expression implies the QNEC for these unitary states, namely $S(t)=S_{(t, +\infty)}(\omega_{V(\rho)} \Vert \omega) $ has positive second derivative. By repeating the computation we are also able to prove that on a generic bounded interval $(a, b)$ we have 
\begin{equation*}
\begin{split}
S_{(a,b)}(\omega_{V(\rho)}  \Vert \omega) & = -\frac{c}{12} \int_{a}^b D_{(a,b)}S\eta(u) du + \frac{c}{12} \log \eta'(a)\eta'(b)-\frac{c}{12}\log \left( \frac{\eta(b)-\eta(a)}{b-a} \right)^2\,,
\end{split}
\end{equation*}
where $S\eta$ is the Schwarzian derivative of $\eta$ and $D_{(a,b)}$ is the density of the dilation operator associated to $(a,b)$. In the case $(a,b)=(-r,r)$ we verify the Bekentein Bound. We also provide the formulas for the relative entropies obtained by exchanging the states $\omega$ and $\omega_{V(\rho)} $. In this case we show a counterexample to the convexity of the function $S_{(t, +\infty)}(\omega \Vert \omega_{V(\rho)} )  $.

%A model of quantum field theory may be in various states: in quantum theory, a system is defined by a the algebra of observables and a physical state is realized as a normalized positive linear functional on it. Among them, the most important is the vacuum state. In this work, we focus on one of the simplest class of quantum field theory: chiral components of two-dimensional quantum field theory. 

\section{Notation and CFT basics}
In this section we describe our notation and some basic facts about the structure of a two-dimensional chiral CFT. The material is standard and more details may be found e.g. in \cite{quantum}. 
The starting point is the {\em Virasoro algebra}, that is the infinite dimensional Lie algebra Vir with generators $\{L_n, c\}_{n \in \Z}$ obeying the relations
\begin{equation} \label{eq:vir}
    [L_n,L_m]=(n-m)L_{n+m}+ \frac{1}{12}n(n^2-1)\delta_{n,-m}c\,,\quad [L_n,c]=0\,.
\end{equation}
A {\em positive energy representation} of Vir on a Hilbert space $\hi$ is a representation such that

    (i) $L_n^*=L_{-n}$,
    
    (ii) $L_0$ is diagonalizable with non-negative eigenvalues of finite multiplicity,
    
    (iii) the central element  is represented by $c\id$. \\

From now, we assume such a positive energy representation on the infinite dimensional separable Hilbert space  $\hi$. We assume furthermore that $\hi$ contains a vector $\Omega$ annihilated by $L_{-1}, L_0, L_{+1}$ ($\mathfrak{sl}(2,\R)$-invariance) which is a highest weight vector of weight $0$, that is $L_n\Omega=0$ for all $n>0$. In \cite{duality}, \cite{FVOAAB}, \cite{cocycle}, \cite{projective} one can find the proof of the bound
\begin{equation} \label{eq:est}
\Vert (1+L_0)^k L_n \Psi \Vert \leq \sqrt{c/2} (|n|+1)^{k+3/2} \Vert (1+L_0)^{k+1}\Psi \Vert
\end{equation}
for $\Psi \in \mathcal{V}= \bigcap_{k \geq 0} \mathcal{D}(L_0^k)$. Given a smooth function $f(z)$ on the circle, one defines the stress energy tensor
\[
T(f)=-\frac{1}{2\pi}\sum_{n=-\infty}^{+ \infty} \left( \int_{S^1}f(z)z^{-n-2}dz \right) L_n\,.
\]
Notice that $T(f) $ has zero expectation on the vacuum, that is $ (\Omega|T(f)\Omega)=0 $. 
This follows by the commutation relations of the Virasoro algebra, since $ L_{-n} \Omega $ is an $n$-eigenvalue of the conformal hamiltonian $L_0$. 
The notation
\[
T(f)= \int_{S^1}T(z)f(z)dz\,, \quad T(z)= -\frac{1}{2\pi}\sum_{n=-\infty}^{+ \infty} z^{-n-2} L_n\,,
\]
is widely used. Moreover, the estimate \eqref{eq:est} shows that $T(f)$ is well defined for any function $f$ in the Sobolev space  $W^{3/2,1}(S^1)$ and that $\mathcal{V}$ is $T(f)$-invariant for any such function. We recall that the norm of $W^{s,p}$ is
\[
\Vert f \Vert_{s,p} = (\sum_n | \hat{f}_n|(1+|n|)^{ps})^{1/p} \,,
\]
where $\hat{f}_n$ is the $n$-th Fourier coefficient. If we now define
\[
\Gamma f(z)=-z^2\overline{f(z)}\,,
\]
then the stress-energy tensor is an essentially self-adjoint operator on any core of $L_0$ (such as $\mathcal{V}$) for any function $f \in W^{3/2,1}(S^1)$ obeying the reality condition
\begin{equation} \label{eq:real}
    \Gamma f = f\,.
\end{equation}
More in general, one has that $T(f)^*=T(\Gamma f)$. We point out that $T(f)$ must be thought of as an operator depending not on the function $f(z)$, but on the vector field $f(z)\frac{d}{dz}$. In particular, we have that
\begin{equation} \label{eq:fields}
    L_n = i T(l_n)\,, \quad l_n= z^{n+1}\frac{d}{dz}\,.
\end{equation}
Notice that by changing variables $z=e^{i \theta}$, the stress energy tensor may be written as
\[
T(f) = \sum_{n}\hat{f}_n L_n \,,
\]
with $f=f(\theta)$. \\

%\begin{lem} 
%$T(f)\Omega=0$ if and only if $f$ is in the Hardy space $H^2(S^1)$, that is if and only if $\hat{f}_n =0$ for every $n<0$. 
%\end{lem}
%\begin{proof} By the relations \eqref{eq:vir} and the hypothesis $L_0 \Omega =0$, one easily obtains that 
%\[
%(L_{-n}\Omega | L_{-n} \Omega) = \frac{c}{12}n (n^2 +1)
%\]
%for any $n>0$. It then follows that
%\[
%\Vert T(f) \Omega \Vert ^2 = -\frac{c}{12} \sum_{n<0} n (n^2 +1)|\hat{f}_{n}|^2 \,,
%\]
%and the thesis is proved.
%\end{proof}

%\begin{lem} 
%Let $\mathcal{H}^{\textit{\em fin}}$ be the dense subspace of {\em finite energy vectors}, that is the algebraic direct sum of the eigenspaces of $L_0$. The vectors $\Psi \in \mathcal{H}^{\textit{\em fin}}$ are entire analytic vectors for the stress energy tensor, that is
%\[
%\sum_{n \geq 0} \frac{\Vert T(f)^n \Psi \Vert}{n!} s^n < + \infty 
%\]
%for any $s>0$.
%\end{lem}
%\begin{proof}
%Given $\Psi \in \mathcal{H}^{\textit{\em fin}}$, by the estimate \eqref{eq:est} it follows that
%\[
%\Vert (1+L_0)^k T(f) \Psi \Vert \leq \sqrt{c/2} \Vert f \Vert_{{k+3/2,1}} \Vert (1+L_0)^{k+1}\Psi \Vert \,.
%\]
%By induction one then obtains that
%\[
%\Vert T(f)^n \Psi \Vert \leq (c/2)^{n/2} \Vert f \Vert^n_{{k+3/2,1}} \Vert (1+L_0)^{n}\Psi \Vert \,.
%\]
%Since the term $\Vert (1+L_0)^{n}\Psi \Vert$ can be estimated by $(1+h)^{n} \Vert \Psi \Vert$, with $h$ the maximum eigenvalue of $L_0$ such that $\Psi$ has a nonzero component in the $h$-eigenspace of $L_0$, the thesis follows.
%\end{proof}

We now make the connection with the representation of the diffeomorphism group on the circle. To do this, given a function $f \in C^\infty (S^1)$ real in the sense of equation \eqref{eq:real}, we denote by $\textit{\em Exp}(tf) = \rho_t \in \textit{\em Diff}_+(S^1) $ the 1-parameter flow of orientation preserving diffeomorphisms generated by the vector field $f$. In other words, $\rho_t$ is uniquely determined by the conditions
\begin{equation} \label{eq:flow}
    \frac{\partial}{\partial t} \rho_t(z) = f(\rho_t(z))\,, \quad \rho_0 = \textit{\em id}\,.
\end{equation}
Notice that $\rho_t$ acts as the identity for all $t \in \R$ outside the support of $f$. The unitary operators $W(f)=e^{iT(f)}$ can be thought of as representers of the diffeomorphisms $\textit{\em Exp}(f)$. More precisely, there exists a strongly continuous unitary projective representation $\Diff \ni\rho \mapsto V(\rho) \in U(\hi)$ satisfying:

V1) $V$ leaves invariant $\mathcal{V}$,

V2) $V$ satisfies the composition law
\[
V(\rho_1)V(\rho_2) = e^{icB(\rho_1, \rho_2)}V(\rho_1 \rho_2)\,,
\]
with $B(\rho_1, \rho_2)$ the {\em Bott 2-cocycle}
\begin{equation} \label{eq:bott}
    B(\rho_1, \rho_2)= -\frac{1}{48 \pi} \textit{\em Re}\int_{S^1} \log(\rho_1 \rho_2)'(z) \frac{d}{dz}\log \rho_2 '(z) dz \,.
\end{equation}

V3) $\frac{d}{dt}V(\textit{\em Exp}(tf))=itT(f)$ on any core of $T(f)$. In particular, we have that $e^{iT(f)}=e^{i\alpha(t)}V(\rho_t)$, with $\alpha'(0)=0$.

We now describe the commutation rules between two operators $e^{iT(f)}$ and $e^{iT(g)}$. For a smooth diffeomorphism $\rho$ on the circle, the {\em Schwarzian derivative} is defined by 
\[
S\rho(z) = \left( \frac{\rho''(z)}{\rho'(z)}\right)' - \frac{1}{2}\left( \frac{\rho''(z)}{\rho'(z)}\right)^2 \,.
\]
It has been shown in \cite{quantum}, which uses results of \cite{cocycle}, \cite{projective} and \cite{laredo}, that on the domain $\mathcal{V}$ we have the relations
\begin{gather}
	 V(\rho)T(g)V(\rho)^* = T(\rho_* g) + \beta(\rho,g) \id\,, \label{eq:transformation} \\
	 i[T(f), T(g)] = T(f'g - g'f) + c \omega(f,g) \id \label{eq:commutation}  \,,
\end{gather}
meaning that \eqref{eq:transformation} holds on $\mathcal{V}$ and \eqref{eq:commutation} on $\mathcal{V} \cap D(T(f)T(g)) \cap D(T(g)T(f))$. Here $\rho_* g$ is the push-forward of the vector field $g(z)\frac{d}{dz}$ through $\rho$ and 
\[
\beta(\rho,g)=-\frac{c}{24\pi}\int_{S^1}g(z)S\rho(z)dz\,, \quad \omega(f,g) = -\frac{c}{48\pi}\int_{S^1} ( f(z)g'''(z)-f'''(z)g(z) )dz \,.
\]
Equation \eqref{eq:transformation} implies that we have the commutation relations 
\[
W(f)W(g)=e^{i\beta(\rho,g)}W(\rho_* g)W(f)\,,
\]
with $W(\cdot)=e^{iT(\cdot)}$ and $\rho= \textit{\em Exp}(f)$. The local net of von Neumann algebras is then given by
\[
\mathcal{A}(I)= \{ W(f)\colon f \in C^\infty_{\mathbb{R}}(S^1)\,, \;  \textit{\em supp}(f) \subset I \}''\,,
\]
with $I$ any open non-dense interval of the circle. 
\begin{defn}
	We will say that a diffeomorphism $\rho$ in $\Diff$ is {\em localized in an interval } $I$ of the circle if $\rho(z)=z$ for each $z \in I'$.
\end{defn}

\begin{lem} \label{lem:haag}
	If $\rho$ is a diffeomorphism localized in $I$, then $V(\rho)$ belongs to $\al(I)$. 
\end{lem}
\begin{proof}
	By duality we can prove that $V(\rho)$ belongs to $\al(I')'$, that is 
	\[
	V(\rho)W(g)V(\rho)^* = W(g)\,, \quad \text{supp} (g) \subset I'\,.
	\]
	But this identity follows from \eqref{eq:transformation}, so the thesis is proved. 
\end{proof}

Going back to the real line, the stress energy tensor $\Theta$ on $\R$ is defined by the formula $ \Theta(f)=T(C_*f) $,
with $C_*f$ the pushforward of the vector field on the real line $f(u)\frac{d}{du}$ through the Cayley trasform $ C(u)= (1+iu)/(1-iu) $. By definition, the stress energy tensor on the real line is then
\begin{equation} \label{eq:real}
\Theta(u)= \left( \frac{dC(u)}{du} \right)^2 T(C(u))= -\frac{4}{(1-iu)^4}T(C(u))\,. 
\end{equation}
%In the real line picture, the commutator between two stress-energy tensors is given by \cite{Tanimoto}
%\[
%[T(f), T(g)] = iT(fg'-f'g) + i \frac{c}{12} \int_\mathbbm{R} f'''(t) g(t) dt \,.
%\]
Using the equations \eqref{eq:real} and \eqref{eq:fields}, we obtain an expression for the generators of $\mathfrak{sl}(2,\R)$. In particular, the generator $D=-\frac{i}{2}(L_1 - L_{-1})$ of dilations is given by
\[
D=\int_{-\infty}^{+\infty}u\Theta(u)du\,.
\]
We recall that, by the Bisognano-Wichmann theorem, we have $D=-2\pi \log \Delta$: the modular dynamic is implemented by boost transformations. \\

We end this section with a few remarks about the representation $V$ of $\Diff$ and the stress energy tensor. Given $n$ points on $S^1$, say $z_i$ with $i = 1, \dots, n$, we denote by $B(z_1, \dots, z_n)$ the group of all the $C^1$ diffeomorphisms $\rho$ of $S^1$ which are smooth except that on the points $z_i$ and such that $\rho(z_i)=z_i$ and $\rho'(z_i)=1$. A similar notation will be used in the real line picture, where in this case we will be particularly interested in $B(\infty)$. If $B$ is the union of all the $B(z_1, \dots, z_n)$ with $z_i \in S^1$, then the unitary projective representation $V$ can be extended to $B$ in such a way that the properties V1)-V3) are still satisfied. For details, see \cite{Hollands}. The relation V3) is then satisfied by any real valued $C^1$ function $f$ on $S^1$ which is smooth except that on a finite number of points $z_i$ such that $f(z_i)=0$ and $f'(z_i)=0$. We precise that if $g$ is a piecewise smooth, real, compactly supported $C^1$-function on $S^1$, then by standard arguments $g \in W^{s,1}$ for any $s<2$. Therefore, $T(g)$ is a closable essentially self-adjoint operator and V3) is verified. Furthermore, in \cite{Carpi} it is proved that if $g_n \to g$ in $W^{3/2, 1}$ then $e^{iT(g_n)} \to e^{iT(g)}$ in the strong operator topology.

The groups $B(z, z')$  are of interest also for the following reason. Given two points $z$ and $z'$ of the circle, consider a diffeomorphism $\rho$ in $B(z,z')$, that is a diffeomorphism in $\Diff$ fixing $z$ and $z'$ and with unital derivative in such points. Define $I=(z,z')$, where the interval is obtained moving counterclockwise from $z$ to $z'$. Then thanks to Lemma \ref{lem:haag} it is possible to define a diffeomorphism $\rho_+$ localized in $I$ and a diffeomorphism $\rho_-$ localized in $I'$ such that $\rho= \rho_+ \rho_- = \rho_- \rho_+$. If $\rho= \text{Exp}(f)$, then this is possible if $f$ and its derivative vanish at the points $z$ and $z'$.

\section{Relative entropy}
We recall the definition of the relative entropy in terms of relative modular operators. We use the notation of $\cite{Entropy}$. 

Let $\mathcal{M}$ be a von Neumann algebra on $\hi$ and let $\varphi$, $\psi$ be two faithful normal positive linear functionals on $\mathcal{M}$ given by the standard vectors $\xi, \eta$. The relative modular operator $\Delta_{\xi, \eta}=S^*_{\xi, \eta}S_{\xi, \eta}$ is given by the polar decomposition 
\[
S_{\xi, \eta}=J\Delta^{1/2}_{\xi, \eta}\,,
\]
with $S_{\xi, \eta}$ the closure of the antilinear operator $X\xi \mapsto X^*\eta$, $X \in \mathcal{M}$. The Connes Radon-Nikodym unitary cocycle can be written as
\[
(D \psi : D \varphi)_t = \Delta_{\eta, \xi}^{it} \Delta^{-it}_{\xi}\,.
\]
The Araki relative entropy is defined by \cite{Araki}
\begin{equation} \label{eq:entropy}
        S(\varphi \Vert \psi)  = -(\xi|\log \Delta_{\eta, \xi} \xi) = - \int_0^{\infty} \log \lambda d(\xi|E_\lambda(\Delta_{\eta, \xi})\xi)\,.
\end{equation}
Notice that
\begin{equation*}
    \begin{split}
        S(\varphi \Vert \psi) 
        & = i \frac{\textit{\em d}}{\textit{\em dt}}\varphi((D \psi : D \varphi)_t) \Big\vert_{t=0}  = -i \frac{\textit{\em d}}{\textit{\em dt}}\varphi((D \varphi : D \psi)_t) \Big\vert_{t=0} \,,
    \end{split}
\end{equation*}
where the last equality follows from the identity $(D \varphi : D \psi)^*_t = (D \psi : D \varphi)_t $. 
It satisfies natural conditions of positivity and monotonocity, that is
\[
S(\varphi \Vert \psi)  \geq 0 \,, \qquad S(\varphi\vert_\mathcal{N} \Vert \psi\vert_\mathcal{N}) \leq  S(\varphi \Vert \psi)\,,
\]
with $\mathcal{N} \subseteq \mathcal{M}$ a von Neumann subalgebra (see \cite{Araki} or \cite{Ohya} to check when the monotonicity property holds). Notice that the relative entropy can be infinite. In this paper we will study in particular the case when $\mathcal{M}= \mathcal{A}(I)$ is a local algebra of the Virasoro net, $\varphi=\omega$ is the vacuum state and $\psi = \omega_U $ is the state $\omega_U(\cdot)=\omega(U^* \cdot U)$ given by the vector $U\Omega$, with $U$ some unitary operator. To this purpose, we give a useful lemma.

\begin{lem} \label{lem:2} Let $\M$ be a von Neumann algebra on $\hi$ and $u \in U(\hi)$ a unitary operator. Consider two standard vectors $\xi$ and $\eta$ for $\M$.

(i) $u\xi$ and $u\eta$ are standard vectors for $u \M u^*$.

(ii) $\Delta_{u\xi, u\eta}^{u \M u^*} = u \Delta_{\xi, \eta}^{ \M } u^*$.

(iii) If $u=vv'$, with $v$ and $v'$ unitary operators in $\M$ and $\M'$ respectively, then $\Delta_{\xi, u\eta}^{\M} = v' \Delta^{\M}_{\xi, \eta} v'^*$.

(iv) If $u=vv'$ as in (iii), then $\Delta_{u \xi, \eta}^{\M} = v \Delta^{\M}_{\xi, \eta} v^*$.
\end{lem}
\begin{proof} (i)
$\overline{u \M u^*(u \xi)}= \overline{\M \xi}$, so $u \xi$ is cyclic and the same holds for $u \eta$. Since the commutant of $u \M u^*$ is $u \M' u^*$, the assertion follows. (ii)
The proof is standard: one first proves that $S_{u\xi, u\eta}^{u \M u^*}=u S_{\xi, \eta} u^*$ and then uses the fact that $\Delta_{\xi, \eta}=S_{\xi, \eta}^* S_{\xi, \eta}$. (iii) In this case we have that $u \M u^* = \M$. The thesis follows by noticing that $S_{\xi, u\eta}=vS_{\xi, \eta} v'^*$ and applying the definition of $\Delta_{\xi, u\eta}$. (iv) The statement follows by (ii) and (iii).
\end{proof}
We now use the previous lemma to do some considerations about the relative entropy between a state $\varphi$ and a state $\psi$.  
Notice that, by the second point of the previous lemma, in general we have
\begin{equation} \label{eq:cov}
    S_{\mathcal{M}}(\varphi \Vert \psi) = S_{U\mathcal{M}U^*}(\varphi_U \Vert \psi_U)\,.
\end{equation}
Moreover, in the particular case $\psi=\varphi_U$ with $U$ as in the lemma, we obtain that the relative entropy \eqref{eq:entropy} has expression
\begin{equation} \label{eq:e2}
S(\varphi \Vert \varphi_U)  = -(V^*\xi|\log \Delta_{\xi} V^*\xi)\,.    
\end{equation}
Similarly, we also have that
\begin{equation} \label{eq:e1}
    \begin{split}
        S(\varphi_U \Vert \varphi) & = - (V\xi|\log \Delta_{\xi} V\xi)\,.
    \end{split}
\end{equation}

These two expressions of the relative entropy will be the starting point of a more explicit formula of the relative entropy.

\section{Relative entropy on a half light-ray}

We now apply the theory of the previous section in (1+1)-dimensional chiral CFT. In particular, the situation is the following: $ \M $ is the local algebra $ \al (0, + \infty) $, $\varphi$ is the vacuum state $\omega$ given by the vacuum vector $\Omega$ and $\omega_U$ is the state associated to $U=V(\rho)$ for some diffeomorphism $\rho$. We recall that if $\rho=\text{ Exp}(f)$ for some real smooth function $f$ then we have $\omega_{W(f)}=\omega_{V(\rho)}$. \\

We now come back to the real line picture and we consider a diffeomorphism $\rho$ in $B(0, \infty)$, so that $\rho(0)=0$ and $\rho'(0)=1$. We also have $\rho(u) \to 0$ and $\rho'(u) \to 1$ if $u \to \infty$. It then follows that $V(\rho)= V(\rho_+)V(\rho_-)$ up to a phase,  where $V(\rho_+)$ belongs to $\al(0, + \infty)$ and $V(\rho_-)$ belongs to $\al(- \infty , 0)$. Notice that the same properties hold for the map 
\begin{equation} \label{eq:eta}
	\eta=\rho^{-1}= \text{Exp}(-f) \,.
\end{equation}
It then follows by the formulas \eqref{eq:transformation} and \eqref{eq:e1} that
\begin{equation*}
    \begin{split}
        S_{(0, +\infty)}(\omega_{V(\rho)} \Vert \omega) 
        & = -\frac{c}{12}\int_{0}^{+ \infty} u S\eta(u)du \,.
    \end{split}
\end{equation*}
Notice that this integral is absolutely convergent, since through the Cayley transform it reduces to an integral of a bounded function on the upper half circle. To prove this one also has to take advantage of the chain rule for the Schwarzian derivative 
\[
S(f \cdot g)(z) = g'(z)^2 Sf(g(z))+Sg(z)\,. 
\]
Therefore, integrating by parts we obtain the expression
\begin{equation} \label{eq:f1}
    S_{(0, +\infty)}(\omega_{V(\rho)} \Vert \omega) = \frac{c}{24}\int_{0}^{+ \infty} u \left( \frac{\eta''(u)}{\eta'(u)}\right)^2 du \,.
\end{equation}
%Similarly, by the formulas \eqref{eq:transformation} and \eqref{eq:e2} we get 
%\begin{equation} \label{eq:f2}
%    S_{(0, +\infty)}(\omega_{V(\rho)} \Vert \omega) = \frac{c}{24}\int_{0}^{+ \infty} u \left( \frac{\rho''(u)}{\rho'(u)}\right)^2 du \,.
%\end{equation}
This formula holds if $\rho \in B(\infty)$ verifies $\rho(0)=0$ and $\rho'(0)=1$. However, by an approximation procedure it is proved in \cite{Hollands} that the same formula holds with the only hypothesis $\rho(0)=0$. The proof is done in the case $\rho = \text{Exp}(f)$ where $f$ has compact support, but it still works with the hypotheses on $\rho$ stated above. The proof is based on the ansatz on the Connes cocycle 
\begin{equation} \label{eq:coc}
    (D \omega_{V(\rho)} : D \omega)= V (\rho_+ \cdot \delta_t \cdot \rho^{-1}_+ \cdot \delta_{-t}) e^{ia(t)}\,, 
\end{equation}
with $\delta_t(u)=e^{-2 \pi t}u$ and $a(t) \in \mathbb{R}$ to be determined. This ansatz has been proved in \cite{Hollands}, but we show it here for the sake of completeness.
\begin{prop}
If $\rho(0)=0$, with $\rho$ in $B(\infty)$, then the equation \eqref{eq:coc} has a solution.
\end{prop}
\begin{proof}
We denote by $u_t$ the right hand side of equation \eqref{eq:coc}, with $a(t)$ to be determined. 
Notice that even though $\rho_+$ is not globally $C^1$, the combination $ \rho_+ \cdot \delta_t \cdot \rho^{-1}_+ \cdot \delta_{-t} $ is globally $C^1$ and so \eqref{eq:coc} is well defined. Moreover the diffeomorphism $\alpha=\rho_+ \cdot \delta_t \cdot \rho^{-1}_+ \cdot \delta_{-t}$ verifies $\alpha(0)=0$ and $\alpha'(0)=1$, so we have that $u_t $ belongs to $\al(0, +\infty)$ for every $t \in \mathbb{R}$. Therefore, the thesis follows if we find a function $a(t)$ such that $u_t$ verifies the relations

(i) $\sigma^t_{\omega_{V(\rho)}}(x)= u_t \sigma^t_\omega(x)u_t^*\,, \quad x \in \al(0, +\infty)\,,$

(ii) $u_{t+s}=u_t\sigma^t_\omega(u_s)\,.$

Note that the first relation suffices to be verified for $x=V(\tau)$, with $\tau=\textit{\em Exp}(g)$ and $\textit{\em supp}\, g \subseteq (0, + \infty)$. Notice also that, since $\rho(0)=0$,  we can apply the point (ii) of Lemma \ref{lem:2}. Therefore, by noticing that $\rho \cdot \delta_t \cdot \rho^{-1} \cdot \tau \cdot \rho \cdot \delta_{-t} \cdot \rho^{-1} = \rho_+ \cdot \delta_t \cdot \rho_+^{-1} \cdot \tau \cdot \rho_+ \cdot \delta_{-t} \cdot \rho_+^{-1}$ and by the explicit expression of the Bott 2-cocycle \eqref{eq:bott} we have
\begin{equation}
\begin{split}
\sigma^t_{\omega_{V(\rho)}}(V(\tau)) & = V(\rho)V(\delta_t)V(\rho)^*V(\tau)V(\rho)V(\delta_t^{-1})V(\rho^*) \\
& = V(\rho_+)V(\delta_t)V(\rho_+^{-1})V(\tau)V(\rho_+)V(\delta_t^{-1})V(\rho_+^{-1}) \\
& = V(\rho_+ \cdot \delta_t \cdot \rho^{-1}_+ \cdot \delta_{-t})V(\delta_t)V(\tau)V(\delta_t^{-1})V(\rho_+ \cdot \delta_t \cdot \rho^{-1}_+ \cdot \delta_{-t})^* \\
& = u_t \sigma^t_\omega(V(\tau))u_t^* \,.
\end{split}
\end{equation}
We now study the condition (ii). This is equivalent to
\begin{equation} \label{eq:sol}
  a(t)+a(s)-a(t+s) = b(t,s) \,,  
\end{equation}
where 
\[
b(t,s)= cB([\rho_+, \delta_t], \delta_t \cdot [\rho_+, \delta_s] \cdot \delta_t^{-1}) \,,
\]
using the usual notation $[g_1, g_2]=g_1 g_2 g_1^{-1}g_2^{-1}$ for the commutator in a group. We can rewrite this condition as {\bf b}a=b, where {\bf b} is the cocycle operator on the additive group $\mathbb{R}$. Since there are not non-trivial 2-cocycles on this group, solutions $a$ of \eqref{eq:sol} can be found provided that {\bf b}b=0, with
\[
{\bf b}b(t,s,r)  = b(t,s)-b(t+s,r)+b(t,r+s)-b(s,r) \,.
\]
By using the identity ${\bf b}B(g_1, g_2, g_3)=0$ one can directly verify this formula (see \cite{Hollands} for the explicit computation).  This concludes the proof.
\end{proof}
\begin{remark} \label{rem}
 Notice that the proposition can be easily adapted to a generic bounded interval $(a,b)$ of the real line, provided that $\rho(a)=a$ and $\rho(b)=b$. 
\end{remark}
We now notice that if $\rho$ is in $B(\infty)$ then $\rho \tau$ is still in $B(\infty)$ for each translation $\tau$. It then follows that, defined $\tau(u)=u+\rho^{-1}(0)$, we have that $\rho\tau$ fixes zero and by the $SL(2,{\mathbb{R}})$-invariance of the vacuum state we also have that $\omega_{V(\rho)}=\omega_{V(\rho)V(\tau)}= \omega_{V(\rho \tau)}$. This implies that the formula \eqref{eq:f1} still holds without the hypothesis $\rho(0)=0$. Now we notice that, by considering $\alpha(t)=u-t$ we have by \eqref{eq:cov} that
\begin{equation} \label{eq:relative}
    \begin{split}
         S_{(t, +\infty)}(\omega_{V(\rho)} \Vert \omega)  = S_{(0, +\infty)}(\omega_{V(\alpha \rho)} \Vert \omega)  =  \frac{c}{24}\int_{t}^{+ \infty} (u-t) \left( \frac{\eta''(u)}{\eta'(u)}\right)^2 du \,.
    \end{split}
\end{equation}

We can also provide a formula for the relative entropy given by exchanging the vacuum and the charged state. Indeed, if as above we use the notation $\eta = \rho^{-1}$, then we have
\begin{equation} \label{eq:sr}
\begin{split}
    S_{(t, +\infty)}(\omega \Vert \omega_{V(\rho)})  = S_{(\eta(t), +\infty)}(\omega_{V(\eta)} \Vert \omega) \,.
\end{split}
\end{equation}
We summarize the previous results in the following theorem.
\begin{thm} \label{thm}
	Let $\rho$ be a diffeomorphism in $B(\infty)$, e.g. $\rho= \textit{\em Exp}(f)$ where $f$ is a real smooth function on the real line such that $f(u) \to 0$ and $f'(u) \to 0$ if $u \to \infty$. Define $\eta = \rho^{-1}$. Then the relative entropies of the correspondent state are finite and given by the formulas
\begin{align*} 
        S_{(t, +\infty)}(\omega_{V(\rho)} \Vert \omega)   & = \frac{c}{24}\int_{t}^{+ \infty} (u-t) \left( \frac{\eta''(u)}{\eta'(u)}\right)^2  du \,, \label{eq:thm}  \\
S_{(t, +\infty)}(\omega \Vert \omega_{V(\rho)})   & = \frac{c}{24}\int_{\rho^{-1}(t)}^{+ \infty} (u-\rho^{-1}(t)) \left( \frac{\rho''(u)}{\rho'(u)}\right)^2  du \,.
\end{align*}
\end{thm}

\begin{cor}
	The states of Theorem \ref{thm} verify the QNEC inequality \cite{QNEC}, namely the function $\prescript{}{\rho}{S}(t) =  S_{(t, +\infty)}(\omega_{V(\rho)} \Vert \omega) $ 
	satisfies  $\prescript{}{\rho}{S}''(t) \geq 0$, since
	\begin{equation*}
	\begin{split}
	\prescript{}{\rho}{S}''(t) & = \frac{c}{24} \left( \frac{\eta''(t)}{\eta'(t)}\right)^2   \,.
	\end{split}
	\end{equation*}
\end{cor}

Now we study the derivatives of  $ S_\rho (t) = S_{(t, +\infty)}(\omega \Vert \omega_{V(\rho)}) $.  Clearly $S_\rho(\rho(t))$ has negative derivative and so $S_\rho(t)$ is decreasing since $\rho$ is increasing. In particular, we have
\begin{equation} \label{eq:der}
    \begin{split}
        S'_\rho (\rho(t))\rho'(t) & = -\frac{c}{24}\int_{t}^{+ \infty} \left( \frac{\rho''(u)}{\rho'(u)}\right)^2 du \,,  \\
        S''_\rho (\rho(t))\rho'(t)^2  & = \frac{c}{24} \left( \frac{\rho''(t)}{\rho'(t)}\right)^2 - S'_\rho(\rho(t))\rho''(t)  \\
        & = \frac{c}{24} \frac{\rho''(t)}{\rho'(t)} \left( \frac{\rho''(t)}{\rho'(t)} + \int_{t}^{+ \infty} \left( \frac{\rho''(u)}{\rho'(u)}\right)^2 du \right) 
        \,.     
    \end{split}
\end{equation}
In this case we can notice that the relative entropy is convex in the average, that is if $\rho = \textit{\em Exp}(f)$ and  $[a,b]$ contains the support of $f$ then
\[
\int_a^b S_\rho''(t)dt = \frac{c}{24}\int_{a}^{b} \left( \frac{\rho''(t)}{\rho'(t)} \right) ^2 du \geq 0 \,,
\]
where this identity follows from \eqref{eq:der} and from the fact that $\rho(u)=u$ outside $[a,b]$. However, in this case the second derivative is not always positive, as shown by the following counterexample. 

\subsection{A counterexample about the second derivative}
Let us consider the function 
\[
f(x) = \left \{ \begin{array}{lr}
\frac{1}{1+ \tan(x)^2} \quad -\pi/2 \leq x \leq \pi/2  \\
0 \qquad \qquad \quad \textit{\em otherwise}
\end{array}
\right. \,.
\]
This is a $C^1$ function with compact support and smooth except that on the points $\pm \pi/2$. We now compute its exponential map $\rho$. Clearly $\rho(u)=u$ outside the interval $[- \pi/2, \pi/2]$, so we can suppose $u \in (- \pi/2, \pi/2)$. Notice that the equation \eqref{eq:flow} can be seen as a family of Cauchy problems
\[
\frac{d}{dt}\rho^u(t)=f(\rho^u(t))\,, \quad \rho^u(0)=u \,,
\]
with $\rho^u(t)=\rho_t(u)$. If $f(u) \neq 0$ then $\rho^u (t) = F_u^{-1}(t)$, with
\begin{equation} \label{eq:F_s(u)}
 F_u(s)= \int_u^s \frac{dv}{f(v)} = \tan (s) - \tan (u) \,.   
\end{equation}
It then follows that
$
\rho_t(u) = F_u^{-1}(t) = \arctan (\tan(u)+t)\,
$
and hence
$
\rho(u)= \arctan (\tan(u)+1)\,.
$
In particular, we have $\rho''(0)/\rho'(0)=-1$. Moreover, by numerical integration one obtains that
\[
\int_0^{\pi/2} \left( \frac{\rho''(u)}{\rho'(u)}\right)^2 du \sim 1.4 \,.
\]
Therefore, by \eqref{eq:der} we obtain
\[
S_\rho''(\pi/4)/4 \sim -c/60 \,,
\]
as announced before. 

\section{Relative entropy on a bounded interval}
In this section we compute the relative entropy between the vacuum state $\omega$ and a unitary state $\omega_{V(\rho)}$ on a bounded interval.\\

First of all we notice that the dilation operator on a bounded interval $I=(a,b)$ can be computed as $D_{(a,b)}= \Theta (D_{(a,b)}(u))$, with
\[
D_{(a,b)}(u) = \frac{1}{b-a}(b-u)(u-a) \,.
\]
Consider a diffeomorphism $\rho$ in $B(\infty)$. We proceed by cases. \\

Suppose that $\rho$ belongs to $B(a,b)$, that is $\rho $ fixes $a$ and $b$ and has unital derivative in such points. As in the case of the half-line we have that $V(\rho)=V(\rho_+)V(\rho_-)$ up to a phase, with $V(\rho_+)$ in $\mathcal{A}(a,b)$ and $V(\rho_-)$ in $\mathcal{A}(a,b)'$. Therefore we can take advantage of the formula \eqref{eq:e2}, and integrating by parts we obtain
\begin{equation*} \label{eq:bounded}
    \begin{split}
        S_{(a,b)}(\omega \Vert \omega_{V(\rho)}) 
        & = \frac{c}{24}\int_a^b D_{(a,b)}(u) \left( \frac{\rho''(u)}{\rho'(u)} \right)^2 du + \frac{c}{12} \int_a^b D'_{(a,b)}(u) \left( \frac{\rho''(u)}{\rho'(u)} \right) du \\
        & = \frac{c}{24}\int_a^b D_{(a,b)}(u) \left( \frac{\rho''(u)}{\rho'(u)} \right)^2 du + \frac{c/6}{b-a} \int_a^b   \log \rho'(u) du \,.
        \end{split}
\end{equation*}

Now we generalize the previous equation to the case in which $\rho'(a)$ and $\rho'(b)$ are generic. Given $r>0$, consider the sequence of functions
\begin{equation} \label{eq:h}
    \begin{split}
        h_n(u)= (n \log r)^{-1} (e^{n (\log r)u}-1)\,.
    \end{split}
\end{equation}

Notice that $h_n(0)=0$, $h_n(1/n) \to 0$ if $n \to + \infty$, $h_n'(0)=1$ and $h'_n(1/n)=r$. Notice also that
\begin{equation} \label{eq:int}
    \int_0^{1/n}u \left( \frac{h''_n(u)}{h'_n(u)} \right)^2 du = \frac{(\log r)^2}{2}\,.
\end{equation}

If we denote the function $\eqref{eq:h}$ by $h_n^r$ and we set $r_a = \rho'(a)$, $r_b= \rho'(b)$ then we can define
\[
h^1_n (u)= a+h_n^{r_a}(u) \,, \quad h^2_n (u) = b- h_n^{r_b}(b-u)\,.
\]
We now consider the following maps: given to intervals $[a,b]$ and $[c,d]$, let $g_{[a,b]}^{[c,d]}(u)=mu+q$ be the affine function mapping $[a,b]$ to $[c,d]$. If $a_n = a + 1/n$ and $b_n = b-1/n$, then we define
\[
g^1_n = g^{[h^1_n(a_n), h^2_n(b_n)]}_{[a,b]}\,, \quad g^2_n = g^{[a, b]}_{[a_n,b_n]}\,.
\]
Finally, we consider the following sequence of functions:
\[
\rho_n(u) = \left \{ \begin{array}{lr}
u \quad \quad u \leq a\, , \; u \geq b \\
h^1_n(u) \quad \quad a \leq u \leq a_n  \\
g^1_n \rho g^2_n(u) \quad \quad a_n \leq u \leq b_n \\
h^2_n(u)  \quad \quad b_n \leq u \leq b \\
\end{array}
\right. \,.
\]
Up to mollify a bit $\rho_n$ in $a_n$ and $b_n$, we have that $\rho_n$ is a sequence of $C^1$ functions such that $\rho_n'(a)=\rho'_n(b)=1$. Moreover, by \eqref{eq:int} one can notice that
\begin{equation} \label{eq:lim}
   \int_a^b D_{(a,b)}(u) S\rho_n (u) du \to - \frac{(\log r_a)^2 + (\log r_b )^2}{4} + \int_a^b D_{(a,b)}(u) S\rho (u) du \,. 
\end{equation}
Now we arrive to the crucial point of the proof. The idea is to approximate $S_{(a,b)}(\omega \Vert \omega_{V(\rho)})$ with $S_{(a,b)}(\omega \Vert \omega_{V(\rho_n)})$, since for the functions $\rho_n$ formula \eqref{eq:bounded} holds. Unfortunately, the relative entropy does not behave well in the limit. However, by studying the Bott 2-cocycle \eqref{eq:bott} it is shown in \cite{Hollands} that $S_{(a,b)}(\omega \Vert \omega_{V(\rho)})$ and $ \lim_n S_{(a,b)}(\omega \Vert \omega_{V(\rho_n)})$ are both solutions of an equation whose solutions are unique up to a constant term $m_\rho$. More precisely, this term depends only on the derivatives $r_a=\rho'(a)$ and $r_b=\rho'(b)$. Therefore by \eqref{eq:lim} we obtain that
\[
S_{(a,b)}(\omega \Vert \omega_{V(\rho)}) = \nu(r_a, r_b) + \frac{c}{24}\int_a^b D_{(a,b)}(u) \left( \frac{\rho''(u)}{\rho'(u)} \right)^2 du + \frac{c/6}{b-a} \int_a^b   \log \rho'(u) du
\]
for some function $\nu(r_a, r_b)$ which we are now going to prove is zero. To do this, we will construct sequences of functions $\rho_n$ with the same derivatives as $\rho$ at $u=a$ and $u=b$. For simplicity we consider $(a,b)= (0,3)$. The general case will follow by covariance, that is by noticing that
\[
S_{(a,b)}(\omega \Vert \omega_{V(\rho)} )  = S_{(0,3)}(\omega \Vert \omega_{V( \alpha \rho \alpha^{-1})} ) \,, 
\]
with $\alpha(u)=cu+d $ in the Moebius group mapping $(0,3)$ in $(a,b)$. 

We start by proving that $0 \leq \nu(r_0, r_3)$.  Given $r>0$, consider the sequence of functions
\begin{equation} \label{eq:sigma}
\sigma_n(u) = \frac{\log n}{\log (n/r)} \left[ (u+ 1/n)^{\log(n/r)/\log(n)} - (1/n)^{\log(n/r)/\log(r)} \right] \,.    
\end{equation}

We notice that $\sigma_n(0)=0$, $\sigma_n(1-1/n)=\frac{\log n}{\log (n/r)} \left( 1- \frac{r}{n} \right) \to 1$ and also $\sigma_n'(1-1/n)=1$. If we denote the function \eqref{eq:sigma} by $\sigma_n^r$, then we define
\[
\rho_n(u) = \left \{ \begin{array}{lr}
\sigma_n^{r_0}(u) \qquad \qquad 0 \leq u \leq 1-1/n \\
3-\sigma_n^{r_3}(3-u) \quad  2+ \frac{1}{n} \leq u \leq 3  \\
\gamma_n(u)  \qquad \qquad \textit{\em otherwise} \\
\end{array}
\right. \,,
\]
with $\gamma_n$ a smooth function such that $\rho_n$ is $C^1$. Moreover, since $\rho_n(1-1/n) \to 1$ and $\rho_n (2+1/n) \to 2$, then we can suppose that $\gamma_n$ converges uniformly with its derivatives (up to the second order) to the identity function  on $[1,2]$. In particular we can suppose that
\[
\int_{1-1/n}^{2+1/n} \rho_n(u) du \to 0 \quad \text{if  } n \to \infty \,.
\]
Therefore, by the positivity of the relative entropy we have
\begin{equation*}
\begin{split}
     0 & \leq \nu (r_0, r_3) + \frac{c}{24} \int_0^3 D_{(0,3)}(u) \left( \frac{d}{du} \log \rho'_n(u) \right)^2 du + \frac{c}{18}\int_0^3 \log \rho'_n(u)du \\
     & \sim \nu(r_0, r_3)  + \frac{c}{24} \left[ \left( \frac{\log r_0}{\log n} \right)^2 \int_0^{1-1/n} \frac{D_{(0,3)}(u)du}{(u+1/n)^4} 
     + \left( \frac{\log r_3}{\log n} \right)^2 \int_{2+1/n}^{3} \frac{D_{(0,3)}(u)du}{(3-u+1/n)^4}\right] \\
     & + \frac{c}{18}\left[ \frac{\log r_0}{\log n}\int_0^{1-1/n}\frac{du}{u+1/n}+ \frac{\log r_3}{\log n}\int_{2+1/n}^{3}\frac{du}{(3-u+1/n)} \right] \\
     & \to \nu(r_0, r_3)\,,
\end{split}
\end{equation*}
where $\sim$ means the equality up to a term going to zero. This proves that $\nu \geq 0$. Now we prove the other inequality.

Given $r>0$, consider
\begin{equation} \label{eq:zeta}
    \zeta_n(u)= - \frac{1}{n}+ \int^u_{-1/n}\exp[(\log r)(ns+1)^{1/n}]ds\,.
\end{equation}
Notice that $\zeta_n(-1/n)= -1/n$, $\zeta'_n(-1/n)=1$ and $\zeta_n'(0)=r$. Notice also that $\zeta_n(0) \to 0$ and $\frac{d}{du}\log \zeta_n'(u) = (\log r) (1+nu)^{1/n-1}$. Always in the case $I=(0,3)$, if we denote the function \eqref{eq:zeta} by $\zeta_n^r$ then we can define
\[
\rho_n(u) = \left \{ \begin{array}{lr}
\zeta^{r_0}_n(u) \qquad \qquad  -1/n \leq u \leq 0 \\
\sigma_n^{r_0}(u)+c_n \quad  0 \leq u \leq 1-1/n  \\
3-\sigma_n^{r_3}(3-u)+d_n \quad  2+1/n \leq u \leq 3 \\
3-\zeta_n^{r_3}(3-u) \quad  3 \leq u \leq 3+1/n \\
\gamma_n(u)  \qquad \qquad \textit{\em otherwise} \\
\end{array}
\right. \,,
\]
with $\gamma_n$, $c_n$ and $d_n$ such that $\rho_n$ is $C^1$. Notice that $c_n \to 0$ and $d_n \to 0$, so that we can suppose that $\gamma_n(u) \to u$ in $[1,2]$ as before. Moreover, if we mollify $\zeta_n$ at $u=0$ then by monotonicity we get
\begin{equation} \label{eq:ineq}
    S_{(0,3)}(\omega \Vert \omega_{V(\rho_n)}) \leq S_{(-1/n , 3+1/n)}(\omega  \Vert \omega_{V(\rho_n)}) \,.
\end{equation}
Notice that on the right side of $\eqref{eq:ineq}$ the term $\nu(r_0, r_3)$ does not compare. Therefore, up to a term going to zero we have
\[
\nu(r_0, r_3) \leq \frac{c}{24} I_n + \frac{c/6}{3+2/n} J_n \,,
\]
with 
\begin{equation*}
    \begin{split}
        I_n & = d_n(\log r_0)^2 \int^0_{-1/n} (1+nu)^{-2(1-1/n)}du  + d_n(\log r_3)^2 \int_3^{3+1/n}(1+n(3-u))^{-2(1-1/n)} du \,, \\
        J_n & = (\log r_0)\int^0_{-1/n} (1+nu)^{1/n}du + (\log r_3)\int_3^{3+1/n}(1+n(3-u))^{1/n} du \,.
    \end{split}
\end{equation*}
But by direct computation and by the estimate $D_{(-1/n, 3+1/n)}(u) \leq u + 1/n$ one has that $I_n \to 0$ and $J_n \to 0$, and so $\nu(r_0, r_3) \leq 0$, as required. \\

We can then conclude with the following formula: if $\rho(a)=a$ and $\rho(b)=b$, then
\begin{align} 
      S_{(a,b)}(\omega \Vert \omega_{V(\rho)}) & = \frac{c}{24}\int_a^b D_{(a,b)}(u)\left( \frac{\rho''(u)}{\rho'(u)} \right)^2 du + \frac{c/6}{b-a}\int_a^b \log \rho'(u) du \label{eq:bfor1} \\
      & = - \frac{c}{12}\int_a^b D_{(a,b)}(u)S\rho(u) du + \frac{c}{12}\log \rho'(a)\rho'(b) \,, \label{eq:bfor2}
\end{align}
where the second expression is obtained by integration by parts. Clearly this formula can be generalized by considering a transformation $\alpha$ in the Moebius group such that $\eta \alpha$ fixes $\eta(a)$ and $\eta(b)$, since in this case we have
\begin{equation} \label{eq:other}
      S_{(a,b)}(\omega_{V(\rho)} \Vert \omega)  =
      S_{\eta(a,b)}(\omega \Vert \omega_{V(\eta)}) = S_{\eta(a,b)}(\omega \Vert \omega_{V(\eta \alpha)}) \,,
\end{equation}
and the formula above applies. \\

Therefore, to generalize the formula \eqref{eq:bfor2} to a generic diffeomorphism $\rho$ in $B(\infty)$, it suffices to explicitly find a diffeomorphism $\alpha$ in the Moebius group such that $\rho \alpha (a)=a$ and $\rho \alpha (b) = b$.  By direct computation one finds the following result.
\begin{thm} \label{thm:bdd}
	Let $\rho$ be a diffeomorphism in $B(\infty)$, for example $\rho = \text{\em Exp}(f)$ where $f$ is a real smooth function on the real line such that $f(u) \to 0$ and $f'(u) \to 0$ as $|u| \to \infty$. If $(a,b)$ is a bounded interval, then
\begin{equation*}
    \begin{split}
S_{(a,b)}(\omega \Vert \omega_{V(\rho)}) & = -\frac{c}{12} \int_{\rho^{-1}(a)}^{\rho^{-1}(b)} D_{\rho^{-1}(a,b)}S\rho(u) du \\ & + \frac{c}{12}\log \rho'(\rho^{-1}(a))\rho'(\rho^{-1}(b))+\frac{c}{12}\log \left( \frac{\rho^{-1}(b)-\rho^{-1}(a)}{b-a} \right)^2\,.        
    \end{split}
\end{equation*}
Similarly, by applying $\eta=\rho^{-1}$ we have that
\begin{equation*}
\begin{split}
S_{(a,b)}(\omega_{V(\rho)} \Vert \omega) & = -\frac{c}{12} \int_{a}^b D_{(a,b)}S\eta(u) du \\ & + \frac{c}{12} \log \eta'(a)\eta'(b)-\frac{c}{12}\log \left( \frac{\eta(b)-\eta(a)}{b-a} \right)^2\,.   
\end{split}
\end{equation*}
\end{thm}

\section{Energy Bounds}

We now apply the previous formulas to verify some conditions expected to hold in QFT.

Consider the state $ \omega \cdot  \Phi$, with $\Phi(\cdot) = V(\rho)^* \cdot V(\rho)$. Notice that the conjugate charge of $\tau (\cdot)= V(\rho) \cdot V(\rho)^*$ is $\Bar{\tau} = \Phi_\rho$: the conjugate equation is trivially satisfied and $d(\tau)=1$. Since the positive generator of translations is
\[
H = \int_{- \infty}^{+ \infty} \Theta(u) du \, ,
\]
then the mean vacuum energy of the charge ${\tau}$ is given by
\begin{equation} 
    \begin{split}
        E & = (\Omega|V(\rho)H V(\rho)^* \Omega) \\
        & = \frac{c}{48 \pi} \int_{- \infty}^{+ \infty}\left( \frac{\rho''(u)}{\rho'(u)}\right)^2 du \,.
    \end{split}
\end{equation}
Before proceeding, we do the following remark: if $\gamma$ is an orientation-preserving diffeomorphism fixing $a$ and $b$, then by the Jensen inequality
\begin{equation} \label{eq:Jensen}
    \begin{split}
\frac{1}{b-a}\int_a^b \log \gamma'(u) du \leq \log \left( \frac{1}{b-a}\int_a^b \gamma'(u) du \right) = 0 \,.
    \end{split}
\end{equation}
We now want to estimate the relative entropy
\[
S(r) = S_{(-r,r)}(\omega \Vert \omega_{V(\rho)})\,.
\]
To this purpose, we notice that if $\alpha(u)=cu+d$ is a Moebius transformation such that $\alpha \rho (\pm r) = \pm r$ then by \eqref{eq:bfor1} we have that
\begin{equation*}
\begin{split}
      S(r) &  \leq \frac{c}{24}\int_{-r}^{r} D_{(-r,r)}(u)\left( \frac{ \rho''(u)}{\rho'(u)} \right)^2 du \,,
\end{split}
\end{equation*}
and so as before we obtain 
\[
S_{(-r,r)}(\omega \Vert \omega_{V(\rho)}) \leq \pi r E\,,
\]
which is related to the Bekenstein Bound \cite{comment}.

Similarly, we now verify the Bekenstein Bound for $\Bar{S}(r)=S_{(-r,r)}(\omega_{V(\rho)} \Vert \omega)$. We denote by $\Bar{E}$ the mean vacuum energy in the representation $\Bar{\tau}$, that is
\begin{equation} \label{eq:en}
    \Bar{E} = (\Omega | V(\rho)^* H V(\rho) \Omega) = \frac{c}{48 \pi} \int_{- \infty}^{+ \infty} \left( \frac{\eta''(u)}{\eta'(u)} \right)^2 du \,.
\end{equation}
If $\alpha(u)=cu+d$ is a transformation in the Moebius group such that $\eta \alpha$ fixes $\eta( \pm r)$, then by \eqref{eq:other} and \eqref{eq:Jensen} we have
\begin{equation}
    \begin{split}
\Bar{S}(r)
& \leq \frac{cr}{48}\int_{-r}^{r} \left( \frac{  \eta ''(u)}{\eta '(u)} \right)^2 du  \,.
    \end{split}
\end{equation}
Therefore, by the identity \eqref{eq:en} we have
\[
\Bar{S}(r) \leq \pi r \Bar{E} \,.
\]
We can also notice that the QNEC \cite{proof} is verified with an equality: if by $S(t)$ we denote the relative entropy \eqref{eq:relative} and
\[
\Bar{E}(t,t')= \frac{c}{24 \pi} \int_{t}^{t'}\left( \frac{\eta''(u)}{\eta'(u)}\right)^2 du \,, \qquad \Bar{E}(t) = \lim_{t' \to t} \frac{E(t,t')}{t'-t}\,,
\]
then we can notice that the QNEC $\Bar{E}(t) \geq S''(t)/2 \pi$ holds with the equality 
\[
\Bar{E}(t)= \frac{1}{2 \pi}S''(t) \geq 0 \,.
\]

\section{A remark about the extensivity} 

It can be easily noticed by the formulas given above that if $S_f=S_I(\omega_f \Vert \omega)$ for some interval $I$, then 
$S_{f_1 + f_2} = S_{f_1}+S_{f_2}$ if the supports of $f_1$ and $f_2$ are disjoint (up to a set of zero measure). Clearly if the supports are not disjoint then this fact is not more true. Therefore, if we define $S_f(t)=S_{(t, + \infty)}(\omega_{V(\rho)} \Vert \omega)$ then we will have 
\[
S_{f_1 + f_2}(t) = S_{f_1}(t) + S_{f_2}(t) + s_t(f_1, f_2)
\]
for some term $s_t(f_1, f_2)$. In this section we give an estimate of $s_t(f_1, f_2)$. \\

Let $ [a_1, b_1] $ and $ [a_2, b_2] $  be the supports of $f_1$ and $f_2$, with $b_1 \leq b_2$. Clearly $s_t(f_1, f_2)=0$ if $t \geq b_1$, since for such values of $t$ we have that $f_1$ does not contribute to the relative entropy. So we can suppose $t \leq b_1$. 

Before proceeding we make a general remark. Consider a real function $f$ with compact support. We recall that if $f(u) \neq 0$ then the exponential flow $\rho_t(u)$ of $f$ is obtained by inverting the function $F_u(s)$ defined in \eqref{eq:F_s(u)}. Then by deriving the relation $t=F_u(\rho_t(u))$ with respect to the variable $u$ we have
\begin{equation} \label{eq:derivative}
	\partial_u \rho_t(u) = \frac{f(\rho_t(u))}{f(u)}\,.
\end{equation}
Deriving again and applying the obtained formulas to $\rho_{-1} = \rho^{-1}= \eta $ and $f=f_1 + f_2$ one obtains
\[
\frac{\eta ''}{\eta '} = \frac{\eta''_1}{\eta'_1} + \frac{\eta''_2}{\eta'_2} + \delta(f_1, f_2)\,,
\]
with 
\[
\delta(f_1 , f_2)(u) = f'(\eta(u))- f'_1(\eta_1(u))-f'_2(\eta_2(u)) \,.
\]
Notice that if $\textit{\em supp}f_1 \cup \textit{\em supp}f_2 \subseteq [a,b]$ then
\[
\Vert \delta(f_1 , f_2) \Vert_{\infty} \leq |b-a|\cdot \Vert f''_1 \Vert_\infty + |b-a|\cdot \Vert f''_2\Vert_\infty  \,.
\]
We use this fact to estimate
\begin{equation*}
\begin{split}
\frac{c}{24}\int_t^{\infty} (u-t) \delta(f_1 , f_2)^2 du & \leq \frac{c}{24} \left( \Vert
f''_1\Vert_\infty + \Vert f''_2 \Vert_\infty \right)^2(b-a)^2\frac{(b_1-t)^2}{2} = \epsilon_0(t) \,. 
\end{split}
\end{equation*}
Moreover, by applying Cauchy-Schwarz with respect to the measure $(u-t)du$ on $(t, +\infty)$ we have
\[
\frac{c}{12}\int_t^{\infty} (u-t) \frac{\eta''_i(u)}{\eta'_i(u)}\delta(f_1,f_2)(u) du \leq \frac{c}{12} \sqrt{S_{f_i}(t) \epsilon_0(t)}= \epsilon_i(t) \,,
\]
for $i=1,2$. Always by Cauchy-Schwarz we have
\[
\frac{c}{12}\int_t^{\infty} (u-t) \frac{\eta_1'' \eta_2'' }{\eta_1' \eta_2'}  du \leq 2 \sqrt{S_{f_1}(t) S_{f_2}(t)} = \epsilon_3(t) \,,
\]
and therefore we can conclude that $|s_t(f_1, f_2)| \leq \epsilon(t)$, with
\[
\epsilon(t)=\epsilon_0(t)+\epsilon_1(t)+\epsilon_2(t)+\epsilon_3(t) \,.
\]
We conclude by noticing that, since $S_{f_1}(t)$ vanishes with its first derivative at $t=b_1$, then $\epsilon(t) \leq C(b_1-t)$ for $t$ near to $b_1$ for some $C>0$.

\section*{Acknowledgements}

I deeply thank Roberto Longo for suggesting me the problem and for the encouragement, and Yoh Tanimoto and Simone del Vecchio for some useful observations.


\begin{thebibliography}{9}

    \bibitem{Araki} 
	H. Araki,
	\textit{Relative Entropy of States of von Neumann Algebras},
	RIMS, Kyoto Univ. 11 (1976), 809-833
	
	\bibitem{proof}
	R. Bousso, Z. Fisher, J. Koeller, S. Leichenauer, A. C. Wall,
	\textit{Proof of the Quantum Null Energy Condition},	
	arXiv:1509.02542 [hep-th]

	
	\bibitem{duality} 
	D. Buchholz and H. Schulz-Mirbach,
	\textit{Haag duality in conformal quantum field theory},
	Rev. Math. Phys. 2, 105 (1990)
	
    \bibitem{FVOAAB} 
	S. Carpi, Y. Kawahigashi, R. Longo and M. Weiner,
	\textit{From vertex operator algebras
    to conformal nets and back},
	Memoirs of the American Mathematical Society 254
    (2018), no. 1213, vi + 85

    \bibitem{Carpi} 
	S. Carpi and M. Weiner,
	\textit{On the uniqueness of diffeomorphism symmetry in conformal field theory},
	Commun. Math. Phys. {\bf 258}, 203 (2005)
	
    \bibitem{QNEC} 
	F. Ceyhan and T. Faulkner,
	\textit{Recovering the QNEC from the ANEC},
	arXiv:1812.04683 [hep-th]
    
	\bibitem{quantum} 
	C.J. Fewster and S. Hollands,
	\textit{Quantum energy inequalities in two-dimensional conformal field theory},
    Rev. Math. Phys. 17, 577–612 (2005)
    
    \bibitem{cocycle} 
	R. Goodman and N.R. Wallach,
	\textit{Structure and unitary cocycle representations of
    loop groups and the group of diffeomorphisms of the circle},
	J. Reine Angew. Math. 347, (1984) 69-133
	
	
    \bibitem{projective} 
	R. Goodman and N.R. Wallach,
	\textit{Projective unitary positive-energy representations
    of $\textit{\em Diff}(S^1)$},
	J. Funct. Anal. 63, (1985), no. 3, 299-321

    \bibitem{Hollands} 
	S. Hollands,
	\textit{Relative entropy for coherent states in chiral CFT},
	arXiv:1903.07508v2 [hep-th]
	
	\bibitem{Lashkari}
	N. Lashkari, H. Liu, S. Rajagopal,
	\textit{Modular Flow of Excited States},
	arXiv:1811.05052 [hep-th]
	
    \bibitem{Longo} 
	R. Longo,
	\textit{Algebraic and modular structure of von Neumann algebras of Physics},
	Proceedings of Symposia in Pure Math. 38, (1982), Part 2, 551
	
	\bibitem{Entropy} 
	R. Longo,
	\textit{Entropy distributions of localised states},
	arXiv:1809.03358v2 [hep-th]

    \bibitem{Landauer}
    R. Longo,
    \textit{On Landauer's principle and bound for infinite systems},
    Comm. Math. Phys. 363, 531-560 (2018)
	 
    \bibitem{comment}
    R. Longo and F. Xu,
    \textit{Comment on the Bekenstein bound},
    J. Geom. Phys. 130, 113 (2018)
    
    \bibitem{Ohya}
    M. Ohya and D. Petz,
    \textit{Quantum entropy and its use},
    Text and Monographs in Physics, Springer-Verlag, Berlin, 1993
    
    \bibitem{laredo}
    V. Toledano Laredo,
    \textit{Integrating unitary representations of infinite-dimensional Lie
    groups},
    J. Funct. Anal. 161 (1999), no. 2, 478-508

	
\end{thebibliography}
\end{document}